\documentclass[sigconf]{acmart}
\usepackage{graphicx}
\usepackage[ruled,vlined,linesnumbered]{algorithm2e}
\usepackage{hyperref}
\usepackage{thm-restate}
\usepackage{thmtools}
\AtBeginDocument{%
  \providecommand\BibTeX{{%
    \normalfont B\kern-0.5em{\scshape i\kern-0.25em b}\kern-0.8em\TeX}}}

\setcopyright{acmcopyright}
\copyrightyear{2018}
\acmYear{2018}
\acmDOI{XXXXXXX.XXXXXXX}

\acmConference[Conference acronym 'XX]{Make sure to enter the correct
  conference title from your rights confirmation emai}{June 03--05,
  2018}{Woodstock, NY}
\acmPrice{15.00}
\acmISBN{978-1-4503-XXXX-X/18/06}

\begin{document}

\title{Space and Move-optimal Arbitrary Pattern Formation on a Rectangular Grid by Robot Swarms}

\author{Avisek Sharma}
\email{aviseks.math.rs@jadavpuruniversity.in}
\orcid{0000-0001-8940-392X}
\affiliation{%
  \institution{Department of mathematics, Jadavpur University}
  \streetaddress{188, Raja S C Mullick Road}
  \city{Kolkata}
  \state{West Bengal}
  \country{India}
  \postcode{700032}
}

\author{Satakshi Ghosh}
\email{satakshighosh.math.rs@jadavpuruniversity.in}
\orcid{0000-0003-1747-4037}
\affiliation{%
  \institution{Department of mathematics, Jadavpur University}
  \streetaddress{188, Raja S C Mullick Road}
  \city{Kolkata}
  \state{West Bengal}
  \country{India}
  \postcode{700032}
}

\author{Pritam Goswami}
\email{pritamgoswami.math.rs@jadavpuruniversity.in}
\orcid{0000-0002-0546-3894}
\affiliation{%
  \institution{Jadavpur University}
  \streetaddress{188, Raja S C Mullick Road}
  \city{Kolkata}
  \state{West Bengal}
  \country{India}
  \postcode{700032}
}

\author{Buddhadeb Sau}
\email{buddhadeb.sau@jadavpuruniversity.in}
\orcid{0000-0001-7008-6135}
\affiliation{%
  \institution{Jadavpur University}
  \streetaddress{188, Raja S C Mullick Road}
  \city{Kolkata}
  \state{West Bengal}
  \country{India}
  \postcode{700032}
}

% \affiliation{%

% \affiliation{%

% \affiliation{%

% \affiliation{%

% \affiliation{%

% \affiliation{%


\begin{abstract}
Arbitrary pattern formation (\textsc{Apf}) is a well-studied problem in swarm robotics. To the best of our knowledge, the problem has been considered in two different settings: one in a euclidean plane and another in an infinite grid. This work deals with the problem in an infinite rectangular grid setting. The previous works in literature dealing with the \textsc{Apf} problem in an infinite grid had a fundamental issue. These deterministic algorithms use a lot of space in the grid to solve the problem, mainly to maintain the asymmetry of the configuration or to avoid a collision. These solution techniques cannot be useful if there is a space constraint in the application field. In this work, we consider luminous robots (with one light that can take three colors) to avoid symmetry, but we carefully designed a deterministic algorithm that solves the \textsc{Apf} problem using the minimal required space in the grid. The robots are autonomous, identical, and anonymous, and they operate in Look-Compute-Move cycles under a fully asynchronous scheduler. The \textsc{Apf} algorithm proposed in \cite{BOSE2020} by Bose et al. can be modified using luminous robots so that it uses minimal space, but that algorithm is not move-optimal. The algorithm proposed in this paper not only uses minimal space but is also asymptotically move-optimal. The algorithm proposed in this work is designed for an infinite rectangular grid, but it can be easily modified to work on a finite grid as well.
\end{abstract}






\maketitle

\section{Introduction}
Swarm robotics, in the field of distributed systems, has been well studied in the past two decades. Replacing a huge, expensive robot with a set of simple, inexpensive robots is the goal of this field. This makes the system cost-effective, robust, and easily scalable. The robot swarm is usually modelled as a collection of computational entities called robots, which can move. These robots operate in Look-Compute-Move (LCM) cycles. In the Look phase, a robot takes a snapshot of its surroundings as input. This input consists of the positions of other robots with respect to their local coordinate system. In the Compute phase, the robot runs an in-built algorithm to determine a position to move to. In the Move phase, the robot moves to that position. The main research interest has been to investigate what minimal capabilities are needed for these robots to solve a problem. The robots are generally assumed to be anonymous (robots have no unique identifiers), autonomous (there is no central control), homogeneous (all robots execute the same distributed algorithm), identical (the robots are indistinguishable from appearance), and disoriented (the robots do not have access to a global coordinate system). Further, if the robots are oblivious (they have no memory to remember their past actions or past configuration) and silent (they have no explicit means of communication), then this robot model is termed the $\mathcal{OBLOT}$ model. Each robot can be equipped with a finite persistent memory, where it can remember a finite bit. In literature, this model is termed $\mathcal{FSTA}$. Each robot can communicate a finite bit of information to other robots. In literature, this model is termed $\mathcal{FCOM}$. The finite bit memory and finite communicable information are together implemented as a finite number of visible lights that can take a finite number of different colors. A robot with visible lights means the robot can access the color of the lights, and other robots can see the lights, which serves as a communication mechanism. This model is termed $\mathcal{LUMI}$ and the robots are called luminous robots. Based on the timing of the activation of the robots and the execution time of the phases of the LCM cycles, there are three types of schedulers in the literature. In the fully synchronous (\textsc{FSync}) scheduler, all robots operate synchronously, where the time is divided into rounds. All robots simultaneously get activated and execute the phases of the LCM cycle. In a semi-synchronous (\textsc{SSync}) scheduler, a nonempty set of robots gets activated in a round and simultaneously executes the phases of the LCM cycle. Next, in the fully asynchronous (\textsc{ASync}) scheduler, there is no common notion of time among the robots. All robots get activated and execute their LCM cycles independently.

The Arbitrary Pattern Formation (\textsc{Apf}) problem is one of the well-studied problems in the literature. This problem asks the robots to form a geometric pattern that is given to them as input. The input is given as a set of points expressed in cartesian coordinates with respect to a coordinate system. The goal of this problem is to design a distributed algorithm that allows a set of autonomous robots to form a specific but arbitrary geometric pattern given as input. This problem has been studied in both the euclidean plane and grid settings. In this paper, the problem is considered on an infinite rectangular grid for luminous robots in a fully asynchronous scheduler. The earlier solutions for this problem in grid settings did not consider the space required for the solution. Infinite grid setting has theoretical motivation, but practically, one cannot have such a luxury. For a space-constrained application field, we need an algorithm that uses less space. It will help to utilise the given space as optimally as possible. This somewhat ensures less total robot movement as well. Motivated by this, this work proposes an algorithm that solves \textsc{Apf} problem in an infinite grid for asynchronous luminous robots. The required space for the proposed algorithm is optimal, and the total number of moves required by the robots is asymptotically optimal. In the next section, we discuss the related works and contributions of this work.
\section{Related work and Our Contribution}
\subsection{Related Work} 

The arbitrary pattern formation problem has been investigated mainly in two settings: one in the euclidean plane and another in the grid. In the euclidean plane, the problem is mainly studied in \cite{BoseDS21,BramasT16,BramasT18,Cicerone19,DieudonnePV10,FlocchiniPSW08,Suzuki96,YAMASHITA10}. In a grid setting, this problem is first studied in \cite{BOSE2020}. Here, the authors solved the problem deterministically on an infinite rectangular grid with oblivious robots in an asynchronous scheduler. Later in \cite{cicerone20}, the authors studied the problem on a regular tessellation graph. Whereas the algorithm proposed in \cite{BOSE2020,cicerone20} is not move-optimal, i.e., the total number of moves made by the robots is not asymptotically optimal. So in \cite{GGSS22}, the authors provided two deterministic algorithms for solving the problem in an asynchronous scheduler. The first algorithm solves the \textsc{Apf} problem for oblivious robots while keeping the total robot movement asymptotically optimal. The second algorithm solves the problem for luminous robots, and this algorithm is asymptotically move-optimal and time-optimal, i.e., the number of epochs (a time interval in which each robot activates at least once) to complete the algorithm is asymptotically optimal. In \cite{KGGSX22}, the authors provided a deterministic algorithm for solving the problem with opaque point robots with lights in an asynchronous scheduler assuming one-axis agreement. Then in \cite{HSVT22}, the authors proposed two randomised algorithms for solving the \textsc{Apf} problem in an asynchronous scheduler. The first algorithm works for oblivious robots. This algorithm is asymptotically move-optimal and time-optimal. The second algorithm works for luminous robots with obstructed visibility (when robots are not transparent). This algorithm is also move-optimal and time-optimal. In \cite{KGGS2022}, the authors solve the problem with opaque fat robots with lights in an asynchronous scheduler assuming one-axis agreement.

\subsection{Space Complexity of APF Algorithms in Rectangular Grid} 
In all the works mentioned for arbitrary pattern formation problems, finding a solution was the first challenge. Then the work tilted towards finding optimal solutions, considering different aspects. So far, the aspects considered were the total number of moves made by the robots and the total time to solve the problem. None of the mentioned works discussed the \textit{space complexity} (Definition~\ref{def0}) of the solution. In \cite{RHth}, the authors considered space complexity, but they showed their solution is asymptotically space optimal. However, in the mutual visibility problem studied in \cite{ABKS22,SVT2020} asymptotic space complexity has been considered.

\begin{definition}\label{def0}
In a rectangular grid, we define the space complexity of an algorithm as the minimum area of the rectangles (whose sides are parallel with the grid lines) such that no robot steps out of the rectangle throughout the execution of the algorithm.
\end{definition}

\paragraph{\bf Space Complexity of earlier APF algorithms and comparison with proposed algorithm} The work proposed in this paper is not only asymptotically space optimal (as in \cite{RHth}), it is exactly space optimal (Theorem~\ref{th2}). Let the smallest enclosing rectangle (SER), the sides of which are parallel to grid lines, of the initial configuration and pattern configuration formed by the robots, respectively, be $m\times n$ and $m'\times n'$. Then the minimum space required for an algorithm to solve the problem is a rectangle of dimension $M\times N$, where $M=\max\{m, m'\}$ and $N=\max\{n,n'\}$. The deterministic algorithm proposed in this paper has space complexity $M\times N$ if $M\ne N$ and space complexity $(M+1)\times N$ if $M=N$. The robots in this work only use one light that can take three different colors. The algorithm proposed in \cite{BOSE2020} can be modified such that it takes up the same amount of space as the algorithm in this work using luminous robots. But the sole technique of the proposed algorithm in \cite{BOSE2020} is not move-optimal. The algorithm proposed in this work is asymptotically move-optimal. The algorithms proposed in \cite{GGSS22} need the robots to form a compact line. The space complexity of these algorithms is $M^2\times N^2$ in the worst case. To the best of our knowledge, the work that is most closely related to our work is \cite{HSVT22}. The first randomised algorithm, proposed in \cite{HSVT22} for luminous non-transparent robots, tends to use less space than all other existing works at this time. But this work did not discuss its spatial complexity. On investigating this work, it appears prima facie that this algorithm uses at least $(M+2)\times(N+2)$ space to execute the algorithm. The authors also did not count the number of lights and colors required for the robots. With a closer look, we observe that this algorithm uses at least 31 colors. The second randomised algorithm for oblivious robots in \cite{HSVT22} has a space complexity of $30M\times30N$. Further, deterministic APF algorithms proposed in \cite{KGGS2022,KGGSX22} solved it for obstructed visibility. These works also need the robots to form a compact line, hence the space complexity of these algorithms is $M^2\times N^2$ in the worst case.

\paragraph{\bf Why do we need an APF algorithm with Optimal Space Complexity}
So far in the full visibility model (where a robot can see all other robots present in the system), the second proposed algorithm in \cite{HSVT22} is best, as it works for the $\mathcal{OBLOT}$ model and is move-optimal, time-optimal as well. Also, the algorithm is deterministic if the initial configuration is asymmetric (definition of asymmetric configuration in Section~\ref{model}). But if we are provided with a $100\times100$ square grid, then the algorithm fails to solve the APF even if the dimension of the SER of the initial and target patterns is $4\times4$. A similar problem arises for the first proposed algorithm in \cite{GGSS22}. Next, visibility becomes poorer if the robots are far away from each other. All the previous works assumed that robots had infinite visibility. But to maintain such an assumption, the robots must be close enough to each other. This can be guaranteed if the space complexity of the algorithm is sufficiently low.

% The first proposed algorithm in \cite{HSVT22} and works in \cite{KGGS2022,KGGSX22} are in weaker model, i,e., in obstructed visibility model but use luminous robots. The algorithm in \cite{HSVT22} is randomised and uses at least 31 different colors. The rest two algorithms are deterministic 

\paragraph*{\bf Our Contribution} This work presents a deterministic algorithm for solving APF in an infinite rectangular grid, which is space-optimal as well as asymptotically move-optimal for the first time. Precisely, the space complexity for the algorithm is $M\times N$ when $M\ne N$ and $(M+1)\times N$ when $M=N$. And, if $\mathcal{D}=\max\{M,N\}$, then each robot requires $O(\mathcal{D})$ moves. The algorithm can be easily modified to work on a finite grid that has enough space to contain both the initial and target configurations. The robots are asynchronous, luminous, and have one light that can take three different colors. To the best of our knowledge so far, this is a deterministic algorithm that has the least space complexity, optimal move complexity, and uses the least number of colors. 
% Due to the space constraints of the manuscript, all proofs are removed from the main paper. These proofs are provided in the Appendix part of the manuscript.

\section{Model and Problem Statement}\label{model}

\paragraph{\bf Robot} 
The robots are assumed to be identical (indistinguishable from appearance), anonymous (no unique identifier), autonomous (no centralised control), and homogeneous (they execute the same deterministic algorithm). The robots are equipped with technology so that a robot can determine the positions of all other robots using a local coordinate system (chosen by the robot). The robots are modelled as points on an infinite rectangular grid graph embedded on a plane. Initially, robots are positioned on distinct grid nodes. A robot chooses the local coordinate system such that the axes are parallel to the grid lines and the origin is its current position. Robots do not agree on a global coordinate system. The robots do not have a global sense of clockwise direction. A robot can only rest on a grid node. Movements of the robots are restricted to the grid lines, and through a movement, a robot can choose to move to one of its four adjacent grid nodes.

\paragraph{\bf Lights} Each robot is equipped with a light that can take three colors, namely, \texttt{off}, \texttt{head}, and \texttt{tail}. A robot can see another robot's light and its present color. Initially, the light of each robot has the same color, \texttt{off}. The colors work as an internal memory as well as a communication technique.

\paragraph{\bf Look-Compute-Move Cycle} 
Robots operate in Look-Compute-Move (LCM) cycles, which consist of three phases. In the Look phase, a robot takes a snapshot of its surroundings and gets the position and color of the lights of all the robots. We assume that the robots have full, unobstructed visibility. In the Compute phase, the robots run an inbuilt algorithm that takes the information obtained in the Look phase and obtains a color (say, $c$) and a position. The position can be its own or any of its adjacent grid nodes. At the end of the compute phase, the robot changes the color to $c$. In the Move phase, the robot either stays still or moves to the adjacent grid node as determined in the Compute phase.

\paragraph{\bf Scheduler} The robots work asynchronously. There is no common notion of time for robots. Each robot independently gets activated and executes its LCM cycle. In this scheduler, the Compute phase and Move phase of robots take a significant amount of time. The time length of LCM cycles, Compute phases, and Move phases of robots may be different. Even the length of two LCM cycles for one robot may be different. The gap between two consecutive LCM cycles, or the time length of an LCM cycle for a robot, is finite but can be unpredictably long. We consider the activation time and the time taken to complete an LCM cycle to be determined by an adversary. In a fair adversarial scheduler, a robot gets activated infinitely often.

\paragraph{\bf Grid Terrain and Configurations} Let $\mathcal{G}$ be an infinite rectangular grid graph embedded on $\mathbb{R}^2$. $\mathcal{G}$ can be formally defined as a geometric graph embedded on a plane as $\mathcal{P}\times \mathcal{P}$, which is the cartesian product of two infinite (from both ends) path graphs $\mathcal{P}$. Suppose a set of robots is placed on $\mathcal{G}$. Let $f$ be a function from the set of vertices of $\mathcal{G}$ to $\mathbb{N}\cup\{0\}$, where $f(v)$ is the number of robots on the vertex $v$ of $\mathcal{G}$. Let $g$ be a function from the set of edges of $\mathcal{G}$ to $\mathbb{N}\cup\{0\}$, where $g(e)$ is the number of robots on the edge $e$ of $\mathcal{G}$. Then the pair $(\mathcal{G},f,g)$ is said to be a \textit{configuration} of robots on $\mathcal{G}$. We assume for the initial configuration $(\mathcal{G},f,g)$, $f(v)=0 \text{ or } 1$ for all nodes $v$ in $\mathcal{G}$ and $g(e)=0$ for all edges $e$. If for a configuration $(\mathcal{G},f,g)$, $g(e)=0$ for all edges $e$, then we call it a \textit{still} configuration. Since for a still configuration $(\mathcal{G},f,g)$, $g$ is fixed, we denote a still configuration as $(\mathcal{G},f)$.

\paragraph{\bf Symmetries} Let $(\mathcal{G},f)$ be a still configuration. A \textit{symmetry} of $(\mathcal{G},f)$ is an automorphism $\phi$ of the graph $\mathcal{G}$ such that $f(v)=f(\phi(v))$ for each node $v$ of $\mathcal{G}$. A symmetry $\phi$ of $(\mathcal{G},f)$ is called \textit{trivial} if $\phi$ is an identity map. If there is no non-trivial symmetry of $(\mathcal{G},f)$, then the still configuration $(\mathcal{G},f)$ is called a \textit{asymmetric} configuration and otherwise a \textit{symmetric} configuration. Note that any automorphism of $\mathcal{G}=\mathcal{P}\times \mathcal{P}$ can be generated by three types of automorphisms, which are translations, rotations, and reflections. Since there are only a finite number of robots, it can be shown that $(\mathcal{G},f)$ cannot have any translation symmetry. Reflections can be defined by an axis of reflection that can be horizontal, vertical, or diagonal. The angle of rotation can be of $90^{\circ}$ or $180^{\circ}$, and the centre of rotation can be a grid node, the midpoint of an edge, or the centre of a unit square. We assume the initial configuration to be asymmetric. The necessity of this assumption is discussed after the problem statement.

\paragraph*{\bf Problem Statement}
Suppose a swarm of robots is placed in an infinite rectangle grid such that no two robots are on the same grid node and the configuration formed by the robots is asymmetric. The Arbitrary Pattern Formation (\textsc{Apf}) problem asks to design a distributed deterministic algorithm following which the robots autonomously can form any arbitrary but specific (target) pattern, which is provided to the robots as an input, without scaling it. The target pattern is given to the robots as a set of vertices in the grid with respect to a cartesian coordinate system. We assume that the number of vertices in the target pattern is the same as the number of robots present in the configuration. The pattern is considered to be formed if the present configuration is a still configuration and is the same up to translations, rotations, and reflections. The algorithm should be \textit{collision-free}, i.e., no two robots should occupy the same node at any time, and two robots must not cross each other through the same edge.  

\paragraph*{\bf Admissible Initial configurations} We assume that in the initial configuration there is no multiplicity point, i.e., no grid node that is occupied by multiple robots. This assumption is necessary because all robots run the same deterministic algorithm, and two robots located at the same point have the same view. Thus, it is deterministically impossible to separate them afterwards. Next, suppose the initial configuration has a reflectional symmetry with no robot on the axis of symmetry or a rotational symmetry with no robot on the point of rotation. Then it can be shown that no deterministic algorithm can form an asymmetric target configuration from this initial configuration. However, if the initial configuration has reflectional symmetry with some robots on the axis of symmetry or rotational symmetry with a robot at the point of rotation, then symmetry may be broken by a specific move of such robots. But making such a move may not be very easy as the robots' moves are restricted to their adjacent grid nodes only. In this work, we assume the initial configuration to be asymmetric.

\section{The Proposed Algorithm}
This section gives the proposed algorithm \textsc{ApfMinSpace}. We assume that the initial configuration formed by the robots is asymmetric and that all of the robots' lights have the color \texttt{off}.

\paragraph{\bf Overview and Key-point of the proposed algorithm}
The proposed algorithm first elects two leader robots that are used to fix the global coordinate system throughout the algorithm. One of the leaders moves to create enough room (if required) for the algorithm to successfully execute. Then non-leader robots move vertically so that each horizontal line contains exactly the number of robots required according to the embedding of the target pattern. Then non-leader robots make horizontal moves to take their respective target positions. Finally, leaders take their respective target positions. Interestingly, with luminous robots, it is not that hard to propose an algorithm that takes minimal space, if we see the technique proposed in \cite{BOSE2020}. But here the algorithm \textsc{ApfMinSpace} takes special care to make the robot move optimally in an optimally bounded space (Fig.~\ref{fig:locus} shows the general locus of any non-leader robot), which leads to making the algorithm space-optimal as well as asymptotically move-optimal.
For a detailed overview of the proposed algorithm, see Appendix.

\subsection{Preliminaries of the Proposed Algorithm}
First, we describe a procedure named Procedure I, which can be executed by a robot if the configuration made by the robot is still and asymmetric. This procedure is used to fix a global coordinate system regardless of the light color of the robots by electing two leaders, namely, head and tail.

\textbf{\textsc{Procedure I:}}

\textit{Assumption}: The current configuration is still and asymmetric.

\textit{Description}: Let $\mathcal{C}=(\mathcal{G},f)$ be the current configuration. Compute the smallest enclosing rectangle (SER) containing all the robots where the sides of the rectangle are parallel to the grid lines. Let $\mathcal{R}=ABCD$ be the SER of the configuration, a $m\times n$ rectangle with $|AB|=n\ge m=|AD|$. The length of the sides of $\mathcal{R}$ is considered the number of grid points on that side. If all the robots are on a grid line, then $R$ is just a line segment. In this case, $R$ is considered a $1\times n$ `rectangle' with $A=D$, $B=C$, and $AD=BC=1$. Let $n>m>1$, that is, $\mathcal{R}$ be a non-square rectangle. For each corner point $A$, $B$, $C$, and $D$ the robot calculates a binary string. For the corner point $A$, the binary string is determined as follows: Scan the grid from $A$ along the longer side $AB$ to $B$ and sequentially all grid lines parallel to $AB$ in the same direction. For each grid point, put a 0 or 1 according to whether it is empty or occupied by a robot. We denote the string as $\lambda_{AB}$ (see Fig.~\ref{fig:lexi}).
\begin{figure}[h!]
    \centering
    \includegraphics[width=.5\linewidth]{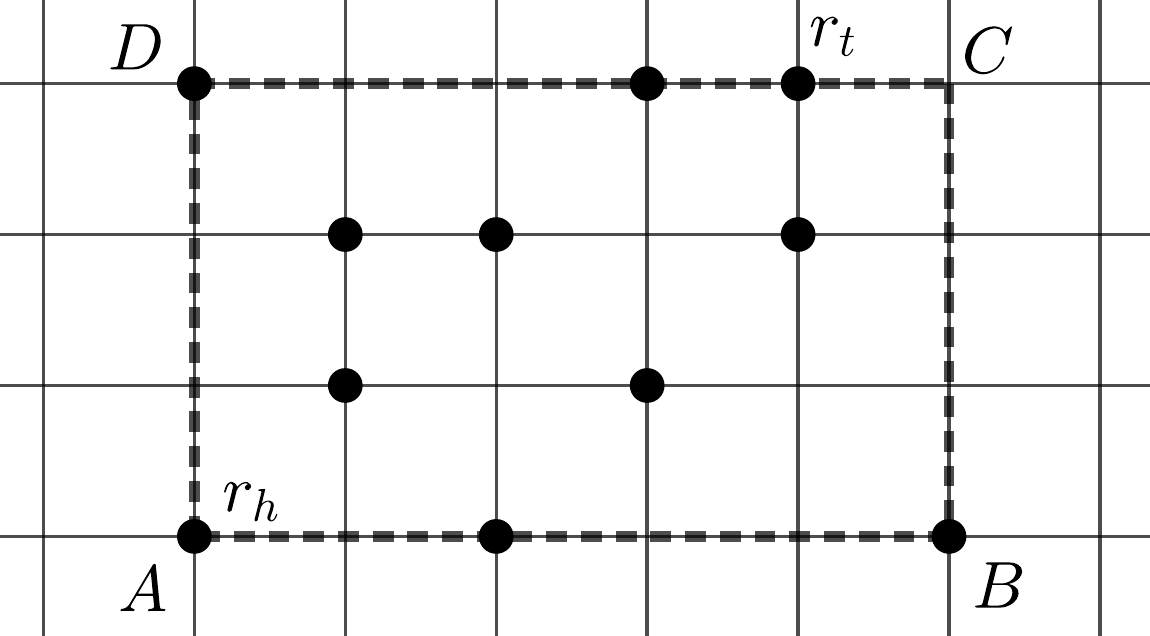}
    \caption{$\lambda_{AB}=101001010100011010100110$ is the largest lexicographic string, and $r_h$ and $r_t$ are respectively the head and tail robots of the configuration. }
    \label{fig:lexi}
\end{figure}
Similarly, the robot calculates the other three strings: $\lambda_{BA}$, $\lambda_{CD}$ and $\lambda_{DC}$. If $\mathcal{R}$ is a square, that is, $m=n$, then we have to associate two strings to each corner. Then we have eight binary strings $\lambda_{AB}$, $\lambda_{BA}$, $\lambda_{AD}$, $\lambda_{DA}$, $\lambda_{BC}$, $\lambda_{CB}$, $\lambda_{DC}$ and $\lambda_{CD}$. Since the configuration is asymmetric, all the strings are distinct. The robot finds the unique lexicographically largest string. Let $\lambda_{AB}$ be the lexicographically largest string, and then $A$ is considered the \textit{leading corner} of the configuration. The leading corner is taken as the origin, and $\overrightarrow{AB}$ is as the $x$ axis and $\overrightarrow{AD}$ is as the $y$ axis. If the $\mathcal{R}$ is a $1\times n$ rectangle, then there are only two associated binary strings: $\lambda_{AB}$ and $\lambda_{BA}$. If both are equal, then the configuration is symmetric. Since the configuration is asymmetric, the strings are distinct. Let $\lambda_{AB}$ be the lexicographically largest string. Then $A$ is considered the origin, and $\overrightarrow{AB}$ is considered the $x$ axis. In this case, there is no common agreement on the $Y$ axis.
In all the cases, a unique string, say $\lambda_{AB}$ is elected. The robot responsible for the first 1 in this string is considered the $head$ robot of $\mathcal{C}$ and the robot responsible for the last 1 is considered the $tail$ of $\mathcal{C}$. The robot other than the head and tail is termed the \textit{inner robot}.

\paragraph*{\bf General definitions of head and tail robots} Let us formally state the definitions of \textit{head} and \textit{tail} robots. Earlier in Procedure I, we defined the head and tail robots. If there is no robot with \texttt{head} color on and the configuration is still and asymmetric, then that definition is applicable. If there is a robot with the light color \texttt{head} and another robot with the light color \texttt{tail}, then the robot with the \texttt{head} color on is said to be the head robot, and the robot with the \texttt{tail} color on is said to be a tail robot. 

Next, we describe another procedure, Procedure~II which directs robots to fix a global coordinate system when there are two robots with respective light colors \texttt{head} and \texttt{tail}.

\textbf{\textsc{Procedure II:}}

\textit{Assumption}: There are two robots with respective light colors \texttt{head} and \texttt{tail}. The head is at a corner of the current SER.

\textit{Description}: Let SER of the configuration be a rectangle $ABCD$ with $|AB|\ge|BC|$ and the head robot situated at $A$. There are three exhaustive cases.
\begin{itemize}
 \item \textit{Case-I}: If $ABCD$ is a non-square rectangle with $|BC|>1$ or if $ABCD$ is a square and the tail robot is on the $CD$ edge but not at $C$, then consider $A$ as the origin, $\overrightarrow{AB}$ as the $x$ axis, and $\overrightarrow{AD}$ as the $y$ axis.
 \item \textit{Case-II}: If $ABCD$ is a square rectangle and the tail robot is at $C$, then consider $A$ as the origin, and there are two possibilities for considering axes. Firstly, it can be done by considering $\overrightarrow{AB}$ as the $x$ axis and $\overrightarrow{AD}$ as $y$ axis. Secondly, it can be done by considering $\overrightarrow{AD}$ as $x$ axis and $\overrightarrow{AB}$ as $y$ axis.
 \item \textit{Case-III}: If the SER of the configuration is a line $AB$ where the head robot is situated at $A$ and the tail robot is situated at $B$, then consider $A$ as the origin and $\overrightarrow{AB}$ as the $x$ axis. The $y$ axis can be considered in either of the two possible ways.
\end{itemize}

\paragraph{\bf Target embedding} Here we discuss how robots are supposed to embed the target pattern when they agree on a global coordinate system. Let the $\mathcal{R}'=A'B'C'D'$ be the SER of the target pattern, an $m'\times n'$ rectangle with $|A'B'|\ge |B'C'|\ge1$. We associate binary strings similarly for $\mathcal{R}'$. Let $\lambda_{A'B'}$ be the lexicographically largest (but may not be unique) among all other strings for $\mathcal{R}'$. The first target position on this string $\lambda_{A'B'}$ is said to be \textit{head-target} and denoted as $h_{target}$ and the last target position is said to be \textit{tail-target} and denoted as $t_{target}$. The rest of the target positions are called \textit{inner target} positions. Then the target pattern is to be formed such that $A'$ is the origin, $\overrightarrow{A'B'}$ direction is along the positive $x$ axis, and $\overrightarrow{A'D'}$ direction is along the positive $y$ axis. Let the SER of the target pattern be a line $A'B'$, and let $\lambda_{A'B'}$ be the lexicographically largest string between $\lambda_{A'B'}$ and $\lambda_{B'A'}$. Then the target is embedded in such a way that $A'$ is at the origin and $\overrightarrow{A'B'}$ direction is along the positive $x$ axis.

Let us define some notations. Let $\mathcal{C'}=\mathcal{C}\setminus\{\text{head}\}$ and $\mathcal{C''}=\mathcal{C}\setminus\{\text{head, tail}\}$, where $\mathcal{C}$ denotes any configuration. Let $\mathcal{C}_{target}'=\mathcal{C}_{target}\setminus\{h_{target}\}$, $\mathcal{C}_{target}''=\mathcal{C}_{target}\setminus\{h_{target},t_{target}\}$ where $\mathcal{C}_{target}$ is the target configuration. Let the dimension of the SER of the current configuration be $m\times n$ with $m\le n$, and the dimension of the SER of the target configuration be $m'\times n'$ with $m'\le n'$. If $m\ge m'$ and $n\ge n'$, then the current SER can contain the target pattern. Next, list some sets of conditions in Table~\ref{tab:0}. Note that a robot can verify these conditions from the current configuration after embedding the target.

% \begin{table}[t]

%     \caption{List of conditions on a configuration $\mathcal{C}$ }
%     \label{tab:0}
%     \scriptsize
%     \centering
%     \begin{tabular}{|p{1.2cm}|p{9cm}|}
%         \hline
%         $C_{asym}$& $\mathcal{C}$ is still and asymmetric\\
%         \hline
%         $C_{final}$  & $\mathcal{C}=\mathcal{C}_{target}$ (final target configuration is achieved) \\
%         \hline
%         $\overline{C}_{h}$  & $\mathcal{C}'=\mathcal{C}_{target}'$ (all target positions are occupied except head-target)  \\
%         \hline
%         $C_{inner}$  & $\mathcal{C}''=\mathcal{C}_{target}''$ (all inner target positions are occupied)\\
%         \hline
%         $C_{off}$  & Light color of each robot in $\mathcal{C}$ is \texttt{off} \\
%         \hline
%         $C_{head}$  & There is a robot at a corner of the SER of $\mathcal{C}$ with a light color \texttt{head} and the light color of the rest robots are \texttt{off} \\
%         \hline
%         $C_{lumi}$ & There is a robot with a light color \texttt{head} at the corner of the SER of $\mathcal{C}$ and there is a robot with a light color \texttt{tail} \\
%         \hline
%         $C_{corner}$ & The tail robot is at a corner point of the SER of $\mathcal{C}$\\
%         \hline
%         $C_{enough}$ & The SER of $\mathcal{C}$ can contain the target pattern\\
%         \hline
%         $C_{rect}$ & The SER of $\mathcal{C}$ is a non-square rectangle\\
%         \hline
%     \end{tabular}
    
% \end{table}

\begin{table}[h!]
    \centering
    \footnotesize
     \caption{List of conditions}
    \label{tab:0}
    \begin{tabular}{|p{1cm}|p{6.5cm}|}
        \hline
        $C_{asym}$& $\mathcal{C}$ is still and asymmetric\\
        \hline
        $C_{final}$  & $\mathcal{C}=\mathcal{C}_{target}$ (final target configuration is achieved) \\
        \hline
        $\overline{C}_{h}$  & $\mathcal{C}'=\mathcal{C}_{target}'$ (all target positions are occupied except head-target)  \\
        \hline
        $C_{inner}$  & $\mathcal{C}''=\mathcal{C}_{target}''$ (all inner target positions are occupied)\\
        \hline
        % $C_{off}$  & Light color of each robot in $\mathcal{C}$ is \texttt{off} \\
        % \hline
        $C_{lumi}$ & There is a robot with a light color \texttt{head} at the corner of the SER of $\mathcal{C}$ and there is a robot with a light color \texttt{tail} \\
        \hline
        $C_{corner}$ & The tail robot is at a corner point of the SER of $\mathcal{C}$\\
        \hline
        $C_{enough}$ & The SER of $\mathcal{C}$ can contain the target pattern\\
        \hline
        $C_{rect}$ & The SER of $\mathcal{C}$ is a non-square rectangle\\
        \hline
    \end{tabular}

 \end{table}

\subsection{Algorithm \textsc{ApfMinSpace}}
Let us formally describe the algorithm \textsc{ApfMinSpace} in Algorithm~\ref{algo:main}. A generic robot $r$, having the light color \texttt{off} initially, runs this algorithm. The algorithm is composed of two main phases. One of them is named phase~\textsc{Lumi}. 
\begin{algorithm}[h!]
\footnotesize
\caption{\textsc{\footnotesize ApfMinSpace}}\label{algo:main} 
    \eIf{$C_{final}$ is true}
    {
        \If{color of robot $r$ is not \texttt{off}}
        {
            robot $r$ changes the color to \texttt{off}\;
        }
    }
    {
        \uIf{$C_{asym}\land \neg C_{lumi}$ is true}
        {
            run phase~\textsc{nonLumi}\;
        }
        \ElseIf{$C_{lumi}$ is a true}
        {
            run phase~\textsc{Lumi}\;
        }
    }
    \end{algorithm}
If $C_{final}$ is not true, a robot infers itself in this phase if $C_{lumi}$ is true. Another phase is named as phase~\textsc{nonLumi}. If $C_{final}$ is not true, a robot infers itself in this phase if $C_{lumi}$ is not true and $C_{asym}$ is true. If $C_{final}$ is true, then a robot changes its light color to \texttt{off}, if it is not already. Next, we describe the phases one by one.

\paragraph*{\textbf{Phase}~\textsc{nonLumi}}~

\textit{Assumption:} $C_{asym}\wedge \neg C_{lumi}$ is true.

\textit{Goal:} $C_{final}\lor (C_{lumi} \wedge \neg \overline{C}_h)$ is true.

\textit{Description}:
Run procedure~I and determine the global coordinate system. If $\overline{C}_h$ is true, then the head goes to the left (right) if the head target is its left (right). If $\overline{C}_h$ is not true, then the tail robot first changes its color to \texttt{tail}. If there is a robot with the color \texttt{tail}, then the head robot starts moving towards the left to reach the origin. When the head robot reaches an adjacent node of the origin, it changes its color to \texttt{head} and moves to the origin. If the head robot is already at origin with color \texttt{off}, then it changes its color to \texttt{head}. The pseudo-code of this phase is given in Algorithm~\ref{algo:ph1}.

\begin{algorithm}[h!]
\footnotesize
\caption{\footnotesize Phase~\textsc{nonLumi}}\label{algo:ph1}   
    run Procedure~I\;
    \eIf{$\overline{C}_h$ is false}
        { 
            \eIf{there is no robot with color \texttt{tail}}
            {
                tail robot changes its color to \texttt{tail}\;
            }
            {
                \uIf{head is at origin with color \texttt{off}}
                {head changes its color to \texttt{head}\;}
                \ElseIf{head is not at origin}
                {
                    \eIf{head is adjacent to origin}
                    {
                        head changes its color to \texttt{head} and moves left\;
                    }
                    {
                        head moves left\;
                    }
                }
            }
        }
        {
            % \eIf{head is at origin}
            % {
            %     \If{color of head is not \texttt{off}}
            %     {
            %         change the color to \texttt{off}\;
            %         move right\;
            %     }
            % }
            
            \eIf{$h_{target}$ is at left}
            {
                head robot moves left\;
            }
            {
                head robot moves right\;
            }
            
        }

\end{algorithm}

\paragraph*{\textbf{Phase}~\textsc{Lumi}}~

\textit{Assumption:} $C_{lumi}$ is true.

\textit{Goal:} $\overline{C}_h\land \neg C_{lumi}$ is true.

\textit{Description}: In this phase, if in the snapshot, the tail robot is seen on the edge, discard the snapshot and go to sleep. Otherwise, run Procedure~II. If considering the coordinate system through case-I or case-III, or any of the coordinate systems through case-II, $C_{inner}$ is true, then the tail moves towards the $t_{target}$ by first moving downwards and then leftwards. When the tail reaches $t_{target}$, head changes its color to \texttt{off} and goes to the right. If $C_{inner}$ is not true considering any of the coordinate systems through Procedure~II, and $C_{corner}$ is false, then the tail robot moves right. If $\neg C_{inner} \wedge C_{corner}$ is true, then there are two possibilities: either $C_{enough}$ is true or not. If $C_{enough}$ is not true, then the tail robot expands the SER to fit the target pattern. If $\neg C_{inner} \wedge C_{corner} \wedge C_{enough}$ is true but $C_{rect}$ is false, then the tail robot moves outside the SER. Finally, when $\neg C_{inner} \wedge C_{corner} \wedge C_{enough} \wedge C_{rect}$ is true, then call the function \texttt{Rearrange}() (this function is described next). The pseudo-code of this phase is given in Algorithm~\ref{algo:ph3}.

\begin{algorithm}[h!]
\footnotesize
\caption{\footnotesize Phase~\textsc{Lumi}}\label{algo:ph3}
\eIf{the tail robot is on an edge}
{   
    do nothing\;   
}
{
    run Procedure~II\;
    \eIf{$C_{inner}$ is true}
    {
        \eIf{$\overline{C}_h$ is false}
        {
            tail robot moves toward $t_{target}$\;
        }
        {
            head robots changes its color to \texttt{off} and moves right\;
        }
    }
    {
        \eIf{$C_{corner}$ is false}
        {
            tail robot moves to right\;
        }
        {
            \eIf{$C_{enough}$ is false}
            {
                tail robot move to expand the SER\;
            }
            {
                \eIf{$C_{rect}$ is false}
                {
                    tail moves outside the SER\;
                }
                {
                    call function \texttt{Rearrage}()
                }
            }
        }
    }
}
\end{algorithm}

% \textit{Goal:} $\overline{C}_h\land \neg C_{lumi}$ is true.

% \begin{wrapfigure}{r}{6.5cm}
%  \begin{algorithm}[H]
%  \scriptsize
% \caption{Phase~III}\label{algo:ph3}
% \eIf{the tail robot is on an edge}
% {   
%     do nothing\;   
% }
% {
%     run Procedure~II\;
%     \eIf{$C_{inner}$ is true}
%     {
%         \eIf{tail is at $t_{target}$}
%         {
%             tail robot changes its color to \texttt{off}\;
%         }
%         {
%             tail robot moves toward $t_{target}$\;
%         }
%     }
%     {
%         \eIf{$C_{corner}$ is false}
%         {
%             tail robot moves right\;
%         }
%         {
%             \eIf{$C_{enough}$ is false}
%             {
%                 tail robot move to expand the SER\;
%             }
%             {
%                 \eIf{$C_{rect}$ is false}
%                 {
%                     tail moves outside the SER\;
%                 }
%                 {
%                     call function \texttt{Rearrage}()
%                 }
%             }
%         }
%     }
% }
% \end{algorithm}
% \end{wrapfigure}

\paragraph*{Function \texttt{Rearrange}()}~

\textit{Assumption:} $C_{lumi}\land\neg C_{inner} \wedge C_{corner} \wedge C_{enough} \wedge C_{rect}$ is true.

\textit{Goal}: $C_{inner}\wedge C_{lumi}$ is true.

\textit{Description}: Let us name the grid lines parallel to the $x$ axis (we shall call them horizontal grid lines): $H_1, H_2, \dots$, from bottom to top, where $H_1$ is the horizontal line that contains the head robot. Let $a'(i)$ $(b'(i))$ be the total number of target positions in $\mathcal{C}_{target}'$ above (below) the $H_i$ horizontal line. Let $a(i)$ $(b(i))$ be the total number of robots in $\mathcal{C}''$ above (below) the $H_i$ horizontal line. We say a horizontal line $H_i$ satisfies \textit{upward} condition if:
\begin{enumerate}
    \item[(U1)] $a'(i)>a(i)$,
    \item[(U2)] [$a'(i+1)>a(i+1)$ \textbf{and} $H_{i+1}$ is empty] \textbf{or} [$a'(i+1)= a(i+1)$].
\end{enumerate}
Next, we say a horizontal line $H_i$ satisfies \textit{downward} condition if:
\begin{enumerate}
    \item[(D1)] $b'(i)>b(i)$,
    \item[(D2)] [$H_{i-1}$ is empty]
    \textbf{or} [$b'(i-1)\le b(i-1)$].
\end{enumerate}
% For the configuration in Figure~\ref{fig:example}, $H_3$ horizontal line satisfies upward conditions. $H_6$ does not satisfy the upward condition, because it does not satisfy (U2). $H_8$ horizontal line satisfies the downward condition. 
A horizontal line $H_i$ is said to be $saturated$ if $a(i)=a'(i)$ and $b(i)=b'(i)$ (See the Fig.~\ref{fig:example}). In this function, inner robots 
move to make $C_{inner}$ true. The algorithm is given below in different cases for a robot $r$ on the horizontal line $H_i$ (for an intuition behind this function, see the Appendix).
\begin{figure}[h!]
    \centering
    \includegraphics[width=.4\textwidth]{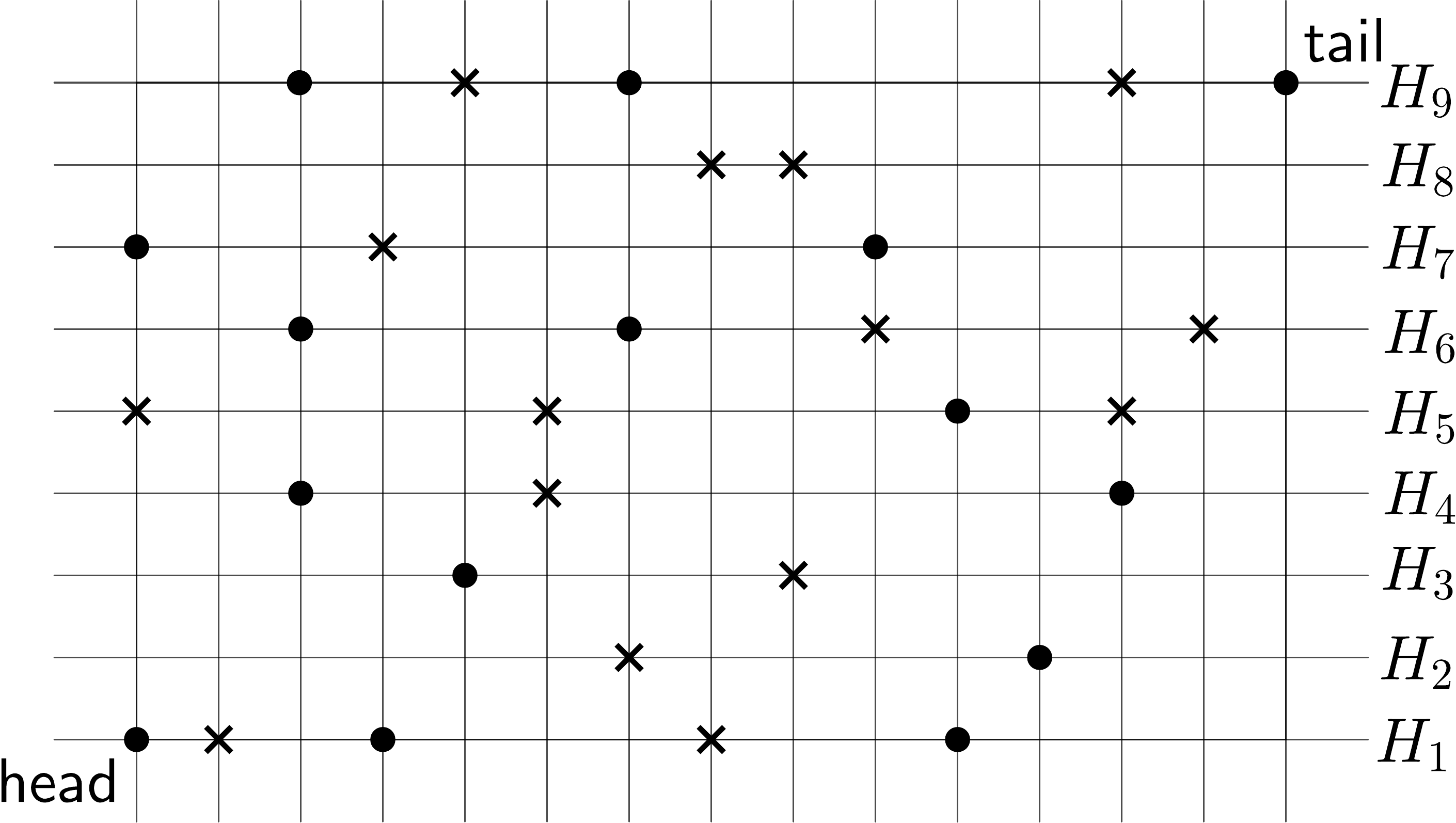}
    \caption{ For this configuration, $b(4)=4, b'(4)=3$ and $a(4)=7, a'(4)=9$; $H_4$ horizontal lines satisfy upward conditions, $H_9$ horizontal lines satisfy downward conditions, and $H_6$ is a saturated horizontal line. Crosses ($\times$) denote target positions.}
    \label{fig:example}
\end{figure}

\textit{Case-I} ($H_i$ satisfies the upward condition but not the downward condition)\\ 
If $H_{i+1}$ is empty, then the leftmost robot on $H_i$ goes upward. If $H_{i+1}$ is nonempty and there is a robot on $H_i$ that has its upward node empty, then the leftmost such robot goes upward. If there are no robots on $H_i$ that have its upward node empty, then consider the leftmost empty node, say, $v$ on $H_{i+1}$ in the SER. The closest robot on $H_i$ from $v$ moves to $v$. If there are two closest robots from $v$ on $H_i$, then the left one moves to $v$.

% \textit{Case-I} ($H_i$ satisfies the upward condition but not the downward condition)\\ If there is a robot in between $H_i$ and $H_{i+1}$ then no robot on $H_i$ does anything. Let there be no robot in between $H_i$ and $H_{i+1}$. If $H_{i+1}$ is empty then the leftmost robot on $H_i$ goes upward. If $H_{i+1}$ is nonempty and there is a robot on $H_i$ which has its upward node empty then the leftmost such robot goes upward. If there are no robots on $H_i$ that have its upward node empty, then consider the leftmost empty node $v$ on $H_{i+1}$ in the SER. Closest robot on $H_i$ to $v$ move to $v$. If there are two closest robots on $H_i$, then the left one moves to $v$. 

\textit{Case-II} ($H_i$ satisfies the downward condition)\\ If $H_{i-1}$ is empty, then the rightmost robot on $H_i$ goes downward. If $H_{i-1}$ is nonempty and there is a robot on $H_i$ that has its downward node empty, then the rightmost such robot goes downward. If there are no robots that have their downward node empty, then consider the rightmost empty node, say, $v$ on $H_{i-1}$ in the SER. The closest robot on $H_i$ from $v$ moves to $v$. If there are two closest robots from $v$ on $H_i$, then the right one moves to $v$.

% \textit{Case-II} ($H_i$ satisfies the downward condition)\\ If there is a robot in between $H_i$ and $H_{i-1}$ then no robot on $H_i$ does anything. Let there be no robot in between $H_i$ and $H_{i-1}$. If $H_{i-1}$ is empty then the rightmost robot on $H_i$ goes downward. If $H_{i-1}$ is nonempty and there is a robot on $H_i$ which has its downward node empty then the rightmost such robot goes downward. If there are no robots that have its downward node empty, then consider the rightmost empty node $v$ on $H_{i-1}$ in the SER. Closest robot on $H_i$ to $v$ move to $v$. If there are two closest robots on $H_i$, then the right one moves to $v$.

% Then consider the closest robot on $H_i$ to $v$ node. If there are two such robots then the left one moves to $v$.

\textit{Case-III} ($H_i$ is a saturated)\\If $r$ is the $j^{th}$ robot from left on $H_i$, then consider the $j^{th}$ target position on $H_i$ from the left at a node, say $t_j$. If the $t_j$ node is at the left (right) of $r$ and the left (right) neighbor node of $r$ is empty, then move left (right).

% \textit{Goal}: $C_{inner}\wedge C_{lumi}$ is true.

% For better understanding phase~III is illustrated in the flowchart depicted in Figure~\ref{fig:fl1}. Next, the flow of the algorithm is given in the flowchart depicted in Figure~\ref{fig:fl2}. 

In the next section, we prove the correctness of the proposed algorithm. The target of the algorithm is to achieve $C_{final}$.

\section{Correctness of the Proposed Algorithm}
We start with the correctness of the three phases. The proofs are omitted from the main paper due to space constraints and provided in the Appendix.
% \paragraph*{Correctness of Phase I}
\begin{lemma}[Correctness of phase~\textsc{nonLumi}]\label{lm1}
If we have a configuration $\mathcal{C}$ in phase~\textsc{nonLumi} at some time $t$, then the following hold true:
\begin{enumerate}
 \item If $\overline{C}_h$ is true at time $t$, then after a finite time $C_{final}$ becomes true.
 \item If $\overline{C}_h$ is not true at time $t$, then after a finite time $\neg \overline{C}_h\land C_{lumi}$ becomes true.
\end{enumerate}

%After finite time execution of phase I, $(C_{off}\wedge C_{final})\lor (C_a\wedge C_{head} \wedge \neg \overline{C}_h)$ becomes true.
\end{lemma}
\begin{proof}
Suppose at time $t$, the configuration is $\mathcal{C}$ in phase~\textsc{nonLumi}. Then $\mathcal{C}$ satisfies $\neg C_{lumi}$ and $C_{asym}$. Then $\mathcal{C}$ is a still and an asymmetric configuration. In this phase, the head robot, selected through Procedure~I, only moves.

(1)~Suppose $\overline{C}_h$ is true at time $t$. In this case, the head moves towards $h_{target}$. If after one move the head reaches $h_{target}$, then $C_{final}$ becomes true and we are done. So, we assume that after one move, the head does not reach $h_{target}$. Let $\mathcal{C}_h$ be the configuration after the move of the head. To show that $\mathcal{C}_h$ is asymmetric and the coordinate system remains unchanged. Suppose at time $t$, $ABCD$ is the SER of the configuration such that $\lambda^{old}_{AB}$ is the lexicographically strictly largest string. There are two possible cases. The $h_{target}$ is either at the left or at the right of the current position of the head.
 
 First, suppose $h_{target}$ is at the left of the current position of head, then head moves left. Since $h_{target}$ is on the $AB$ line segment, if after the move, the head reaches at $A$, then $h_{target}$ is at $A$. This contradicts our assumption, so after the move, the new position of the head is not at $A$. If $i^{th}$ term of $\lambda^{old}_{AB}$ is the first nonzero term, then $(i-1)^{th}$ term in $\lambda_{AB}^{new}$ is 1. But the $(i-1)^{th}$ term in all other considered strings is zero in $\mathcal{C}_h$. Thus, $\lambda_{AB}^{new}$ is the strictly largest in $\mathcal{C}_h$. Thus, $\mathcal{C}_h$ remains asymmetric and the coordinate system remains unchanged.
 
 Next, suppose $h_{target}$ is at the right of the current position of the head. According to the target embedding, the SER of the embedded target pattern should also be $ABCD$ with $\lambda^{target}_{AB}$ as a lexicographically largest string, and the position of the $h_{target}$ is the first 1 in $\lambda^{target}_{AB}$. Let $i^{th}$ term of the $\lambda^{target}_{AB}$ in the target embedding be the first nonzero term. Until the head reaches $h_{target}$, it remains at the left of the $h_{target}$. Suppose $j^{th}$ term of $\lambda^{new}_{AB}$ in $\mathcal{C}_h$ is the first nonzero term. Then $j<i$. Since except head, each robot is occupying their respective target positions, so $\lambda^{new}_{AB}$ strictly larger string in $\mathcal{C}_h$. Thus, $\mathcal{C}_h$ remains asymmetric and the coordinate system remains unchanged.
 
 After each head movement towards $h_{target}$ the distance between them decreases. Therefore, after a finite time, the head reaches at $h_{target}$ and $C_{final}$ becomes true.
 
(2)~Suppose $\overline{C}_h$ is false at time $t$. In this case, the head robot moves left until it reaches the origin. But first, the tail robot is elected through procedure~I changes its color to \texttt{tail} on activation if it is not already. Let $\mathcal{C}_h$ be the configuration after the move of the head. Let $r$ be the head robot at time $t$. Let $ABCD$ be the SER of the configuration at time $t$ with $\lambda_{AB}$ as strictly largest string. Suppose, by making one move towards the left, the head does not reach the corner of the SER. We already proved in (i) that if the head moves left but does not reach a corner, the configuration remains asymmetric and the coordinate system remains unchanged. Hence, after a finite time, the head will reach the adjacent node of $A$ on the $AB$ line segment. Next time, when the head gets activated, it changes its color to \texttt{head} and moves to $A$. This makes $\neg \overline{C}_h\land C_{lumi}=$ true. 
\end{proof}
% \begin{proof}
% See Section~\ref{L1proof} in Appendix.   
% \end{proof}

% \paragraph*{Correctness of Phase II}
% \begin{lemma}[Correctness of Phase~II]\label{lm2}
% If we have a configuration $\mathcal{C}$ in phase~II at some time $t$, then the following hold true: 
% \begin{enumerate}
%     \item If $C_{asym}$ is not true at time $t$, then after finite time $\neg C_{asym}\land C_{off}$ becomes true.
%     \item If $C_{asym}$ is true at time $t$, then after finite time $C_{asym}\land((C_{off} \land \overline{C}_h)\lor (C_{lumi}\land \neg \overline{C}_h))$ becomes true.
% \end{enumerate}
% \end{lemma}  
% \begin{proof}
%  In this phase, no robots are directed to move. The correctness of this phase is basic.   
% \end{proof} 

% \paragraph*{Correctness of Phase III}
\begin{lemma}[Correctness of \texttt{Rearrange}()]\label{lemmaR}
    If we have a configuration $\mathcal{C}$ that satisfies $C_{lumi}\land\neg C_{inner} \wedge C_{corner} \wedge C_{enough} \wedge C_{rect}$ is true at some time, then through out the execution of the function \texttt{Rearrange} the coordinate system does not change and after a finite time $C_{inner}\land C_{lumi}$ becomes true.
\end{lemma}
\begin{proof}
Since only inner robots are moving in this function and no inner robot is allowed to step out of the SER formed by the head and tail robots, so $C_{corner}$, $C_{rect}$ and $C_{enough}$ remain true throughout the execution of \texttt{Rearrange}(). So the coordinate system decided through Procedure II also remains unchanged throughout.

If a nonempty horizontal line becomes saturated, then after a finite time, all robots on that line take their respective target positions by horizontal moves. It can be shown that while this horizontal movement continues, no collision or deadlock situation will occur. To show that $C_{inner}$ will become true within finite time, it is sufficient to show that within finite time all the horizontal lines will become saturated.

If the SER of the initial and target configurations are both lines, then the SER has only one horizontal line. This horizontal line is vacuously saturated. Next, suppose SER in the current configuration has more than one horizontal line. At some point, let there be a non-saturated horizontal line in the configuration. Let us consider two saturated horizontal lines, $H_i$ and $H_j$, such that $|i-j|\ne 1$ and all the horizontal lines in between them are non-saturated. If there is no saturated horizontal line or only one saturated horizontal line, then consider this scenario in the following way: Let the SER of the configuration be $ABCD$, where head and tail are respectively situated at $A$ and $C$. Then consider the horizontal line below $AB$ and the horizontal line above $CD$. We can consider these two lines as vacuous saturated lines. The scheme of the proof is that we show that after a finite time, another saturated horizontal line will form between the lines $H_i$ and $H_j$. Without loss of generality, let $i>j$. Note that there have to be at least two horizontal lines between $H_i$ and $H_j$. Consider the horizontal line $H_{i-1}$. Note that $H_{i-1}$ cannot satisfy the upward condition because $a'(i-1)=a(i-1)$. According to the assumption, $H_{i-1}$ is not a saturated horizontal line. So we must have $b'(i-1)>b(i-1)$ or $b'(i-1)<b(i-1)$.

\textit{Case I:} ($b'(i-1)>b(i-1)$) Let starting from $H_{i-1}$ and going downwards $H_k$ be the last horizontal line such that $b'(k)>b(k)$ and $k>j+1$. Then $b'(k-1)\le b(k-1)$ and $b'(p)>b(p)$ for all $p=i-1, i-2, \dots, k$. The existence of such a horizontal line is guaranteed because $b'(j+1)= b(j+1)$. Consider the horizontal line $H_k$. Then $H_k$ satisfies the downward condition. If $H_k$ is nonempty, then a robot will come down. $H_k$ horizontal line will keep satisfying the downward condition until $b'(k)=b(k)$ becomes true. Suppose $H_k$ is empty, then consider the first nonempty horizontal line $H_m$ above $H_k$. Then $H_m$ satisfies the downward condition. Then a robot comes down from $H_m$. That robot comes down to $H_k$ making $H_k$ nonempty. Hence, after a finite time, $b'(k)=b(k)$ becomes true. Now if at this time $b'(k+1)=b(k+1)$ then $H_k$ is saturated and our task is done. Since only robots were coming down through $H_{k+1}$ in this time interval, therefore $H_{k+1}$ was satisfying (D1), so the difference $b'(k+1)-b(k+1)$ can minimum reach zero. So we have only one remaining possibility at this time, that is, $b'(k+1)>b(k+1)$. Suppose $b'(k+1)>b(k+1)$, then now $H_{k+1}$ satisfies the downward condition. And similarly, after finite time, $b'(k+1)=b(k+1)$ becomes true, which implies that $H_k$ is saturated. 

\textit{Case II:} ($b'(i-1)<b(i-1)$) For this case, we have $a'(i-2)>a(i-2)$ and we have $a'(i-1)=a(i-1)$. Hence, $H_{i-2}$ satisfies the upward condition. If $H_{i-2}$ also satisfies the downward condition, then $H_{i-2}$ must be nonempty, and after a finite time, the required robot(s) will go down from $H_{i-2}$, making it no longer satisfy the downward condition. We assume $H_{i-2}$ satisfies the upward condition but not the downward condition. If $H_{i-2}$ is nonempty then a robot goes upward and reaches $H_{i-1}$. The $H_{i-2}$ will keep satisfying the upward condition, and the robot will keep coming up from $H_{i-2}$ until $H_{i-2}$ is empty or $a'(i-2)=a(i-2)$. If $a'(i-2)=a(i-2)$ turns true, then $H_{i-1}$ becomes saturated. If $a'(i-2)>a(i-2)$ is true and $H_{i-2}$ is empty, then consider the first nonempty horizontal line $H_m$ below $H_{i-1}$. Note that such a nonempty line must exist. Consider $H_{i-3}$ horizontal line. We have $a'(i-2)>a(i-2)$ is true and $H_{i-2}$ is empty. This forces it to satisfy $a'(i-3)>a(i-3)$. So, $H_{i-3}$ satisfies the upward condition. Similarly, we can show that all horizontal lines $H_{i-2},\dots, H_m$ satisfy the upward condition. Since $H_{m}$ satisfies the upward condition, a robot comes upward from $H_m$. And that robot reaches $H_{i-2}$ making $H_{i-2}$ non-empty. Hence, after a finite time, $H_{i-1}$ becomes saturated.
\end{proof}

\begin{lemma}[Correctness of phase~\textsc{Lumi}]\label{lm3}
If we have a configuration $\mathcal{C}$ in phase~\textsc{Lumi} at some time $t$, then after a finite time $ \overline{C}_h\land \neg C_{lumi}$ becomes true.
\end{lemma}
\begin{proof}
If at time $t$, $C_{inner}$ is true, the tail moves towards $t_{target}$. If tail is not at the horizontal line containing $t_{target}$, then it moves downwards. Note that, throughout this move, the coordinate system does not change because the larger side of the SER remains larger. After a finite number of moves downward, the tail reaches the horizontal line that contains $t_{target}$. On reaching the same horizontal line as $t_{target}$, tail moves left until it reaches $t_{target}$. While moving left, there can be a time when the SER of the configuration is a square where the head and tail robots are at opposite corners. In this scenario, there could be two possible coordinate systems, according to Procedure~II. But with respect to one of them, $C_{inner}$ will be true, and that coordinate system will be considered by the tail robot. Except for this, when the tail robot moves left to reach $t_{target}$, there will be no ambiguity or change in the coordinate system. After a finite number of moves towards the left, the tail reaches $t_{target}$ resulting in $\overline{C}_h=$ true. At this time, head robot is at origin with color \texttt{head}. On the next activation of the head robot, if the $C_{final}$ is not true, then the head robot turns its color to off and moves right. This results in $\overline{C}_h\land \neg C_{lumi}=$ true.

Next, if $C_{inner}$ is not true, then we show that after a finite time, $C_{inner}$ becomes true. If $C_{inner}$ is not true, then there are three types of moves by the tail robot. Firstly, when $C_{corner}$ is not true, the tail moves right. Since throughout this move the tail does not reach the corner, there is no ambiguity regarding the coordinate system according to Procedure~II and the coordinate system also does not change. Next, when $C_{enough}$ is not true, tail moves outside the SER. Since the target pattern has a finite dimension, after a finite move outside the SER, $C_{enough}$ becomes true. Next, if $C_{rect}$ is not true, then the tail robot moves one step outside the SER. After one move, the $C_{rect}$ becomes true. Hence, after a finite time, $C_{lumi}\land C_{corner}\land C_{enough}\land C_{rect}$ becomes true. From the Lemma~\ref{lemmaR}, after finite time, $C_{inner}$ becomes true. The flow of this phase is depicted in Fig.\ref{fig:fl1}.
\end{proof}

\begin{figure}[h!]
    \centering
    \includegraphics[width=.47\textwidth]{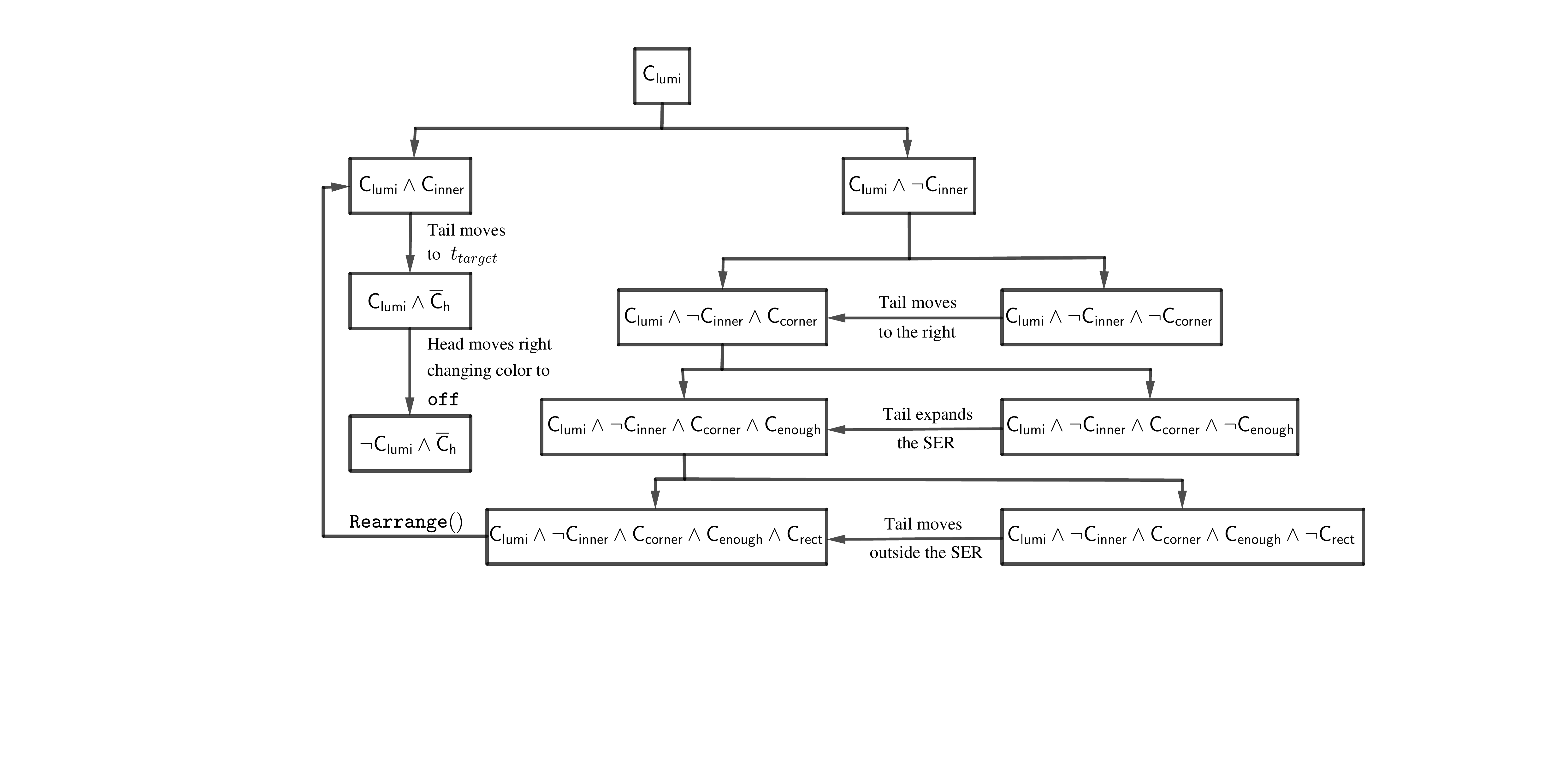}
    \caption{Illustration of phase~\textsc{Lumi}}
    \label{fig:fl1}
\end{figure}

%     \begin{proof}
% See Section~\ref{L3proof} in Appendix.  
% \end{proof}

Next, we prove the correctness of the Algorithm~\ref{algo:main} (proof is given in Appendix).
% Next, we prove the correctness of the proposed algorithm. The goal of our algorithm is to make a configuration where $C_{off}\land C_{final}$ is true. If the initial configuration doesn't match with the target configuration then $\neg C_{final}\land C_{off} \land C_a$ is true. The flow chart depicted in Figure~\ref{fig:fl2} shows that any directed path starting from $\neg C_{final}\land C_{off} \land C_a$ ends at $C_{final}\land C_{off}$ passing through the phases finite times. Hence the correctness follows. 

\begin{theorem}\label{th0}
If the initial configuration is asymmetric, then a set of $k$ asynchronous robots can form any pattern consisting of $k$ points in finite time by executing the algorithm \textsc{ApfMinSpace}.
\begin{proof}
If for the initial configuration $C_{final}$ is not true, then according to our assumption $C_{asym}\land \neg C_{lumi}$ is true. Therefore, initially, the algorithm enters into phase~\textsc{nonLumi}. The algorithm is correct if, within a finite time, robots move such that $C_{final}$ becomes true. 

If initially $\overline{C}_h$ is true, then from Lemma~\ref{lm1}, it results in $C_{final}=$ true within a finite time. If initially $\overline{C}_h$ is not true, from Lemma~\ref{lm1}, it results in $C_{lumi}\land \neg\overline{C}_h=$ true within a finite time. Then the algorithm enters phase~\textsc{Lumi}. From Lemma~\ref{lm3}, after finite time, $\overline{C}_h\land C_{lumi}$ becomes true. Suppose, at time $t_1$, $\overline{C}_h\land C_{lumi}$ becomes true. Let $\mathcal{C}_1$ be the configuration at time $t_1$. If at this point $C_{final}$ is not true, then in phase~\textsc{Lumi}, head changes its color to \texttt{off} and moves right, which results in $\overline{C}_h\land \neg C_{lumi}=$ true. Let, after the move of the head towards the right, the configuration become $\mathcal{C}_2$. Note that the SER, say, $ABCD$, is the same for both configurations $\mathcal{C}_1$ and $\mathcal{C}_2$. Let us name the head robot in $\mathcal{C}_1$ as $r$.

% Consider the time $t_1$. Let $\mathcal{C}_1$ be the configuration at time $t_1$. Since $\mathcal{C}_1\ne \mathcal{C}_{target}$, according to the target embedding $h_{target}$ is at the right of the head in $\mathcal{C}_1$. Let $ABCD$ be the SER of $\mathcal{C}_1$ where head is situated at $A$. Then, $ABCD$ is also the SER of the target embedding and $\lambda_{AB}^{target}$ is lexicographically largest string in $C_{target}$. Since $h_{target}$ is not at $A$, all the corners of SER $ABCD$ is unoccupied in $C_{target}$. Since in $\mathcal{C}_1$, all target positions are occupied except $h_{target}$, so in $\mathcal{C}_1$ all corner points except $A$ is unoccupied. Thus, in $\mathcal{C}_1$ there cannot have rotational symmetry or, horizontal or vertical reflectional symmetry. Thus, $\mathcal{C}_1$ can only have diagonal symmetry if $ABCD$ is a square. Suppose $\mathcal{C}_1$ has a diagonal symmetry, then the line of symmetry is the line $AC$. When head robot moves right, then the symmetry is broken and after the move new configuration is asymmetric.

\begin{figure}[h!]
    \centering
    \includegraphics[width=.3\textwidth]{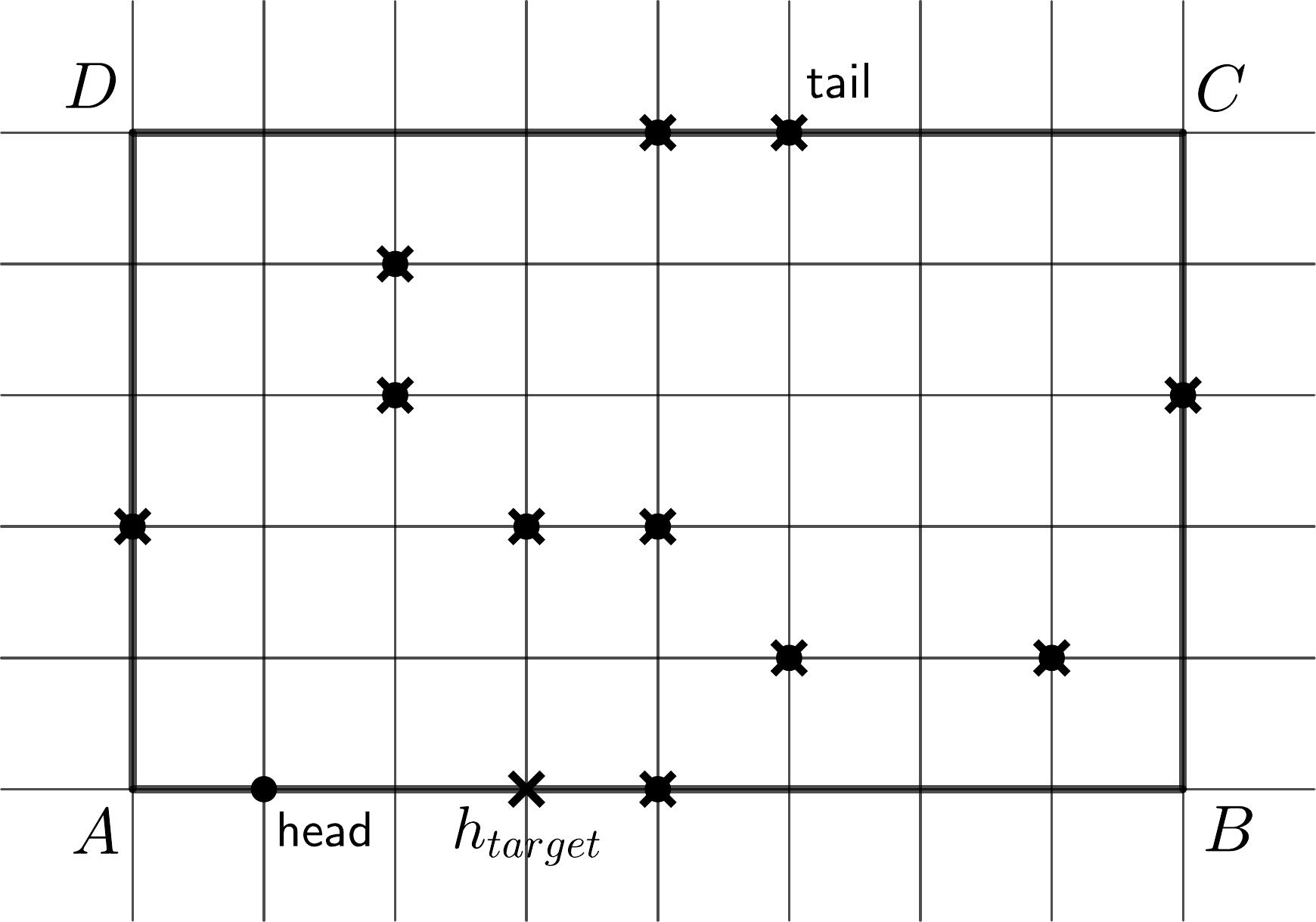}
    \caption{An image related to Theorem~\ref{th0}}
    \label{fig:th}
\end{figure}
If $\mathcal{C}_2\ne \mathcal{C}_{target}$, then the robot $r$ is closer than $h_{target}$ to $A$. So, $\lambda_{AB}$ in $\mathcal{C}_2$ is lexicographically strictly a larger string than $\lambda^{target}_{AB}$ because all the target positions are occupied except $h_{taregt}$ in $\mathcal{C}_2$ (See Fig.~\ref{fig:th}). Thus, $\lambda_{AB}$ in $\mathcal{C}_2$ is lexicographically strictly larger string, making $\mathcal{C}_2$ an asymmetric configuration. Therefore, after the finish of phase~\textsc{Lumi} the resultant configuration satisfies $\overline{C}_h\land \neg C_{lumi}\land C_{asym}$. Thus, the algorithm again enters in phase~\textsc{nonLumi} with $\overline{C}_h=$ true. From Lemma~\ref{lm1}, after a finite time, $C_{final}$ becomes true. The flow of the algorithm is depicted in Fig.~\ref{fig:fl2}.
\end{proof}
\end{theorem}
\begin{figure}[h!]
    \centering
    \includegraphics[width=0.25\textwidth]{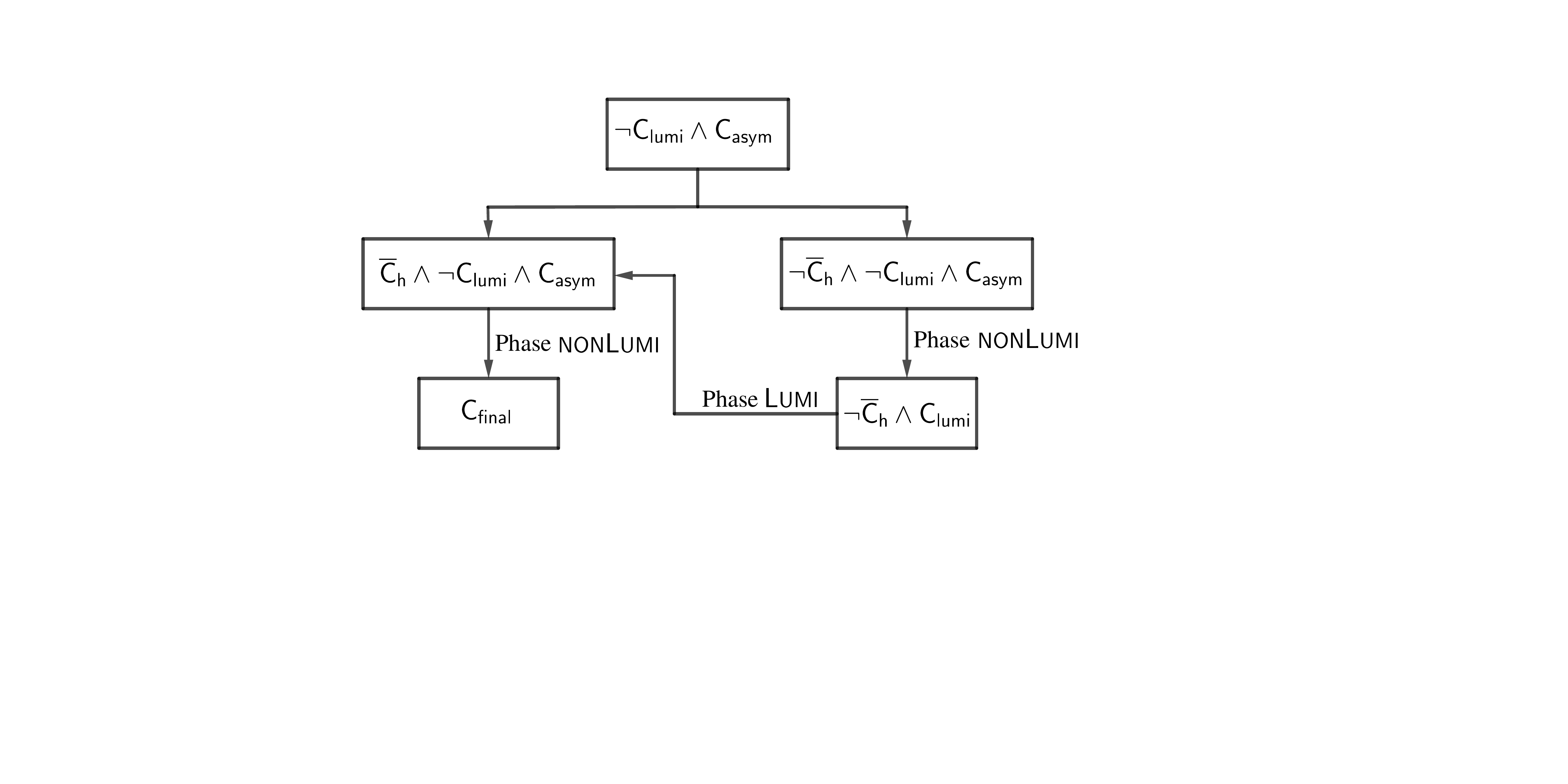}
    \caption{\footnotesize Flow of the Algorithm \textsc{ApfMinSpace}}
    \label{fig:fl2}
\end{figure}

\section{Space and Move complexity of the Proposed Algorithm}
In this section, in Theorem~\ref{th2} (proof is provided in the Appendix), we first calculate the maximum space required for the robots to execute the algorithm \textsc{ApfMinSpace}. Then we calculate the move complexity of the algorithm in Theorem~\ref{th3} (proof is provided in the Appendix). The move complexity of an algorithm is formally defined as the total number of moves made by the robots throughout the execution of the algorithm. Let $\mathcal{D}$ be the dimension of the minimum square, which can contain both the initial and the target configuration.

\begin{theorem}\label{th2}
Let $m\times n$ ($m\ge n$) and $m'\times n'$ ($m'\ge n'$) be the dimensions of the SER of the initial configuration and target configuration, respectively. Let $M=\max\{m,m'\}$ and $N=\max\{n,n'\}$. Then, throughout the execution of algorithm \textsc{ApfMinSpace}, the robots are only required to move inside a rectangle of dimension $M\times N$ or $(M+1)\times N$ in accordance with $M>N$ or $M=N$.
\end{theorem}
\begin{proof}
We try to find the robots that move out of the current SER at any moment. They take up more space. In phase~\textsc{nonLumi}, the head robot only moves and maximum reaches at the leading corner of the current SER, which means that the head robot never steps out of the current SER while being in this phase. In phase~\textsc{Lumi}, head robot only moves inside the SER. An inner robot only moves in \texttt{Rearrange} function in this phase, where they are also not allowed to step out of the current SER. In this phase, the tail robot only steps out of the current SER if it has to expand the SER to contain the target pattern. Hence, if the SER of the initial configuration can contain the target pattern, then no robot steps out of the SER of the current configuration. Otherwise, the tail expands the SER exactly to fit the target pattern. Hence, the robots only move inside a rectangle with minimum dimensions that contains both the initial and target configurations. Precisely, the robots only move inside a rectangle of dimension $M\times N$. Further, if $M=N$ then the tail moves one step away from the current SER to make the SER a non-square rectangle in phase~\textsc{Lumi}. Hence we get the maximum required space as stated in the Theorem~\ref{th2}.
\end{proof}
% \begin{proof}
%  See Section~\ref{T2proof} in Appendix.
% \end{proof}

\begin{theorem}\label{th3}
The algorithm \textsc{ApfMinSpace} requires each robot to make $O(\mathcal{D})$  moves.
\end{theorem}
\begin{proof}
First, we consider the movements of the head and tail robots. The head robot only moves through the $x$ axis, and its maximum locus is from initial position to origin and then origin to head-target. Hence, head maximum makes $2\mathcal{D}$ moves. The tail robot might change its horizontal line to reach the horizontal line that contains $t_{target}$ but only once, then it moves to the $t_{target}$. Therefore, the tail makes at most $2\mathcal{D}$ moves.

Next, let $r$ be an inner robot that initially belonged to the horizontal line $H_i$. If $H_i$ is saturated, then $r$ makes maximum $\mathcal{D}$ moves in \texttt{Rearrange} function to reach its respective target position. Suppose $H_i$ is not saturated initially. Firstly, we make an observation. If for a horizontal line $H_i$, $a'(i)\le a(i)$ ($b'(i)\le b(i)$), then no robot ever goes upward (downward) from $H_i$. If $a'(i)\le a(i)$ is true, then it implies and is implied by $b'(i+1)\ge b(i+1)$. If $a'(i)=a(i)$ then it implies and is implied by $b'(i+1)= b(i+1)$. So at this condition, neither $H_{i}$ satisfies the upward condition nor $H_{i+1}$ satisfies the downward condition. So, there will be no exchange of robots between these two horizontal lines. Suppose $a'(i)< a(i)$. Then this implies and is implied by $b'(i+1)> b(i+1)$. Then eventually leads to satisfying the downward condition for $H_{i+1}$. Let $l=b'(i+1)-b(i+1)$, then from the proposed algorithm, a unique fixed robot on $H_{i+1}$ robot comes down from $H_{i+1}$ to make the difference $b'(i+1)-b(i+1)=l-1$. If $l-1>0$, there is another fixed unique robot that comes down from $H_{i+1}$ making the difference $l-2$. Hence, eventually the difference $b'(i+1)-b(i+1)$ becomes zero. After this, no exchange of robots between the horizontal lines $H_{i}$ and $H_{i+1}$ takes place. Thus, if $a'(i)\le a(i)$ is true, then no robot ever goes upward from $H_i$. Similarly, one can show that, if for a horizontal line $H_i$, $b'(i)\le b(i)$, then no robot ever goes downward from $H_i$.

Thus, if $a'(i)\le a(i)$ and $b'(i)\le b(i)$, $r$ never leaves the $H_i$. Eventually $H_i$ gets saturated, so in this case also $r$ makes at most $\mathcal{D}$ moves. Otherwise, after some time, either it satisfies the upward condition, the downward condition, or both. In this case, either $r$ never leaves the horizontal line or $r$ goes upward or downward. Suppose $r$ goes upward, then at that time $a'(i)>a(i)$ and either $H_{i+1}$ is empty or $a'(i+1)=a(i+1)$. First, we show that if $a'(i+1)=a(i+1)$ is true, then $r$ never leaves $H_{i+1}$ after reaching there. If $a'(i+1)=a(i+1)$, then from the previous discussion, no robot ever goes up from $H_{i+1}$. After $r$ moves upward, if $a'(i)=a(i)$ becomes true, then $H_{i+1}$ is saturated. Otherwise, if $a'(i)>a(i)$ remains true, then $b'(i+1)<b(i+1)$. Then $H_{i+1}$ does not satisfy the downward condition. Hence, $H_{i+1}$ does not satisfy either the upward or downward condition. So then $r$ does not leave the $H_{i+1}$ after that.

Next, suppose $H_{i+1}$ is empty and $a'(i+1)>a(i+1)$. If $r$ moves upward, then it just makes one vertical movement to reach $H_{i+1}$. Hence, we conclude that if $r$ goes upward from $H_i$ under the condition $a'(i+1)=a(i+1)$ then it $r$ settles down at a target position on $H_{i+1}$. And if $r$ goes upward from $H_i$ under the condition that $H_{i+1}$ is empty, then it just makes a vertical movement to reach $H_{i+1}$. A similar conclusion can be drawn if $r$ starts moving downwards in the first place.

% Similarly, if after reaching $a'(i)=a(i)$ becomes true then $H_{i+1}$ becomes saturated.

Next, we show that throughout the execution of the algorithm, if $r$ starts moving upward, then it never comes down after that, and if it starts moving downward initially, then it never goes upward after that. On the contrary, let the opposite happen. Then, without loss of generality, there exists $i$ such that $r$ goes upward from $H_i$ to $H_{i+1}$ and then, after some time, again comes down. $r$ goes upward, it implies $a'(i)>a(i)$ and any other robot can go upward after that only if $a'(i)>a(i)$ remains true. After a robot goes upward, we must have $a'(i)\ge a(i)$. That implies $b'(i+1)\le b(i+1)$, but with this $H_{i+1}$ can never satisfy the downward condition. So no robot can come downwards from $H_{i+1}$.

\begin{figure}[ht!]
    \centering
    \includegraphics[width=.4\textwidth]{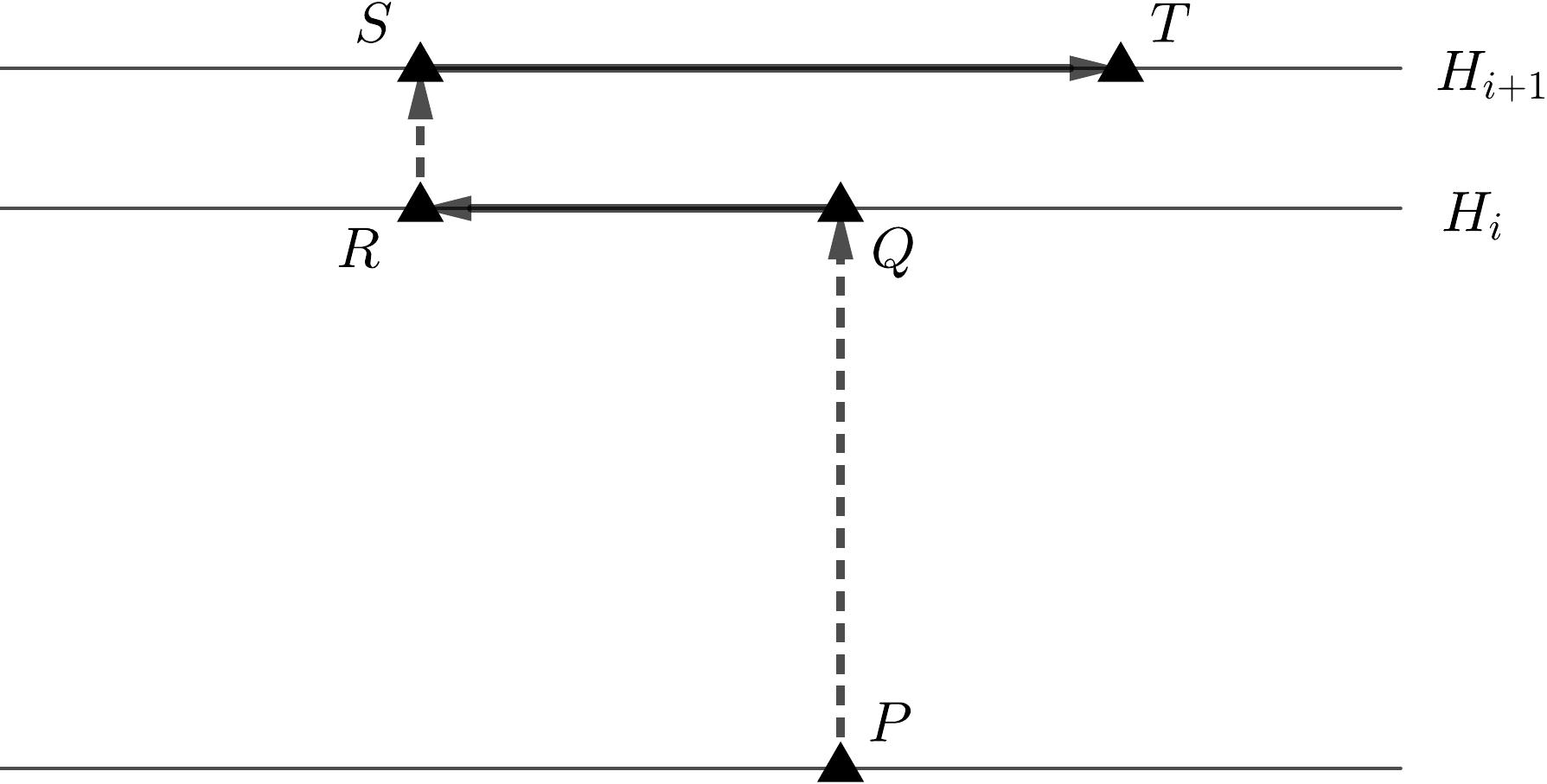}
    \caption{\footnotesize Locus of an inner robot going upward starting from $P$ and settling down at the target position $T$}
    \label{fig:locus}
\end{figure}

Summarising all, if a robot $r$ starts initially going upward, its locus would be a vertical movement in a straight line until it reaches a horizontal line $H_i$ where $a'(i+1)=a(i+1)$. This takes at most $\mathcal{D}$ moves. Then $r$ maximum makes $\mathcal{D}$ moves on the $H_{i}$ horizontal line before moving to $H_{i+1}$. Then $r$ moves upward on $H_{i+1}$ and settles down at a target node on $H_{i+1}$, which also takes at most $\mathcal{D}$ moves (See Fig.~\ref{fig:locus}). Hence, a robot $r$ makes $3\mathcal{D}$ moves in total. Similarly, one can show that if $r$ starts moving downwards initially, then also $r$ all total makes $3\mathcal{D}$ moves.
\end{proof}
% \begin{proof}
%  See Subsection~\ref{T3proof} in Appendix.
% \end{proof}

If the number of robots present in the configuration is $k$, the total required move for the proposed algorithm is $O(\mathcal{D}k)$. In \cite{BOSE2020} authors proved that, any algorithm solving the APF problem requires $\Omega(\mathcal{D}k)$ moves. This shows that the proposed algorithm is asymptotically move-optimal.
 
\section{Conclusion}

This work provided an algorithm for solving the arbitrary pattern formation problem by robot swarms. The robots are considered autonomous, anonymous, and identical. The proposed algorithm works for asynchronous robots with one light that can take three different colors. The algorithm uses minimal space to solve the \textsc{Apf} problem (Theorem~\ref{th2}). Further, the algorithm is asymptotically move-optimal (Theorem~\ref{th3}). Even though the proposed algorithm is considered over an infinite rectangular grid, the algorithm can be easily modified to work on a finite rectangular grid (see Appendix for the preliminary idea) if the dimension of the grid is large enough as required by Theorem~\ref{th2} (this part shall be discussed in a detailed version of the work).

This work does not investigate (due to space constraints) whether the algorithm is asymptotically time-optimal or not. If the proposed algorithm is not time-optimal, then it would be interesting to find out whether there exists an algorithm that is asymptotically move-optimal, time-optimal, and also space-optimal. Further, in the proposed algorithm, for a case where $M=N$, the algorithm requires the space $(M+1)\times N$. We do not know whether this can be improved to $M\times N$ but it is under process. Even though the proposed algorithm is asymptotically move-optimal, we believe that the total required move is better than existing move-optimal \textsc{Apf} algorithms (which shall be investigated in a detailed version). Next, this work uses luminous robots, but it will be interesting to find the lower bound of the space complexity of \textsc{Apf} algorithms for $\mathcal{OBLOT}$ robot model.

\bibliographystyle{ACM-Reference-Format}
\bibliography{sample-base}

\appendix
 \newpage
\begin{center}
\LARGE\textbf{APPENDIX}
\end{center}

\section{An overview of the proposed algorithm}
\paragraph{Procedures to find leaders and fix global coordinate system} The algorithm calls two procedures named Procedure~I and Procedure~II. Both procedures are called to find the two leader robots, head and tail, and to fix the global coordinate system. The Procedure~I is called when $C_{lumi}$ is not true and the configuration is asymmetric. For an asymmetric configuration, it is always possible to find a unique global coordinate system. The procedure~II is called when $C_{lumi}$ is true. The configuration can be symmetric while calling this procedure, but the presence of different colors of robots breaks the symmetry. If the SER of the current configuration is a square and the head and tail robots are at opposite corners, then there are two possible ways to consider the global coordinate system. 

\paragraph{Phase~\textsc{nonLumi}} A robot infers itself in this phase when the configuration, visible, is asymmetric and $C_{lumi}$ is not true. This phase either terminates the algorithm by making $C_{final}$ true when $\overline{C}_h$ is true, otherwise it makes $C_{lumi}$ true. In this phase, if the tail robot's color is \texttt{off}, then it changes its color to \texttt{tail}. Then head robot moves either towards the target when $\overline{C}_h$ is true or towards the origin to make $C_{lumi}$ true.

\paragraph{Phase~\textsc{Lumi}} This phase first makes the tail robot move to the corner of the current SER opposite the head robot. Then it expands the SER enough so that it contains the target pattern. Then the tail robot moves outside the SER to make the SER a non-square rectangle, if required. Then the algorithm calls the function \texttt{Rearrange}(). This function aims to make $C_{inner}$ true, that is, all inner robots are occupying their respective target positions. We discuss this function in the next paragraph. After $C_{inner}$ is true, the tail moves to its tail-target $t_{target}$. This makes $\overline{C}_h=$ true. Then the head changes its color to \texttt{off} moves right. This last move of the head makes $\neg C_{lumi}=$ true and leaves the configuration asymmetric, if $C_{final}$ is not already true, to allow the algorithm to enter into phase~\textsc{nonLumi}.

\paragraph{Function~\texttt{Rearrange}()} In this function, the coordinate system is determined by Procedure~II and maintained by the unchanged positions of the head and tail robots throughout the execution of the function. This function aims to relocate robots such that each horizontal grid line contains exactly the number of robots that are required on that line according to the target embedding. When for a horizontal line these conditions are satisfied, we call it a saturated line. In this case, no more robot exchanges will take place on this line. In such a case, from case-III of this function, robots on this line move horizontally to take their respective target positions. If a horizontal line is not saturated, then there needs to be some exchange of robots through this line. We say there is a scarcity of robots above a horizontal line $H_i$ if the number of robots present above $H_i$ (this number is denoted as $a(i)$) is less than the total number of target positions above $H_i$ (this number is denoted as $a'(i)$). If there is a scarcity above $H_i$, i.e., $a'(i)>a(i)$, then robots are supposed to move up from $H_i$. This is (U1) in upward condition. But to avoid collision, we cannot let a robot go upward only based on scarcity. So we bring another condition. If $a'(i+1)>a(i+1)$ is also true at the time, then robots are supposed to move upward from $H_{i+1}$ also. But no robot should come on $H_{i+1}$ from above because that would increase the scarcity above. But if $H_{i+1}$ is empty at this point, then it is safe for a robot to move upward from $H_i$. So we have one condition in (U2): $a'(i+1)>a(i+1)$ and $H_{i+1}$ is empty. Another alternative condition in (U2) is that we have $a'(i+1)=a(i+1)$. In this case, there is no scarcity of robots above $H_{i+1}$ and no extra robot above either. So there will should not be any exchange of robots from above $H_{i+1}$. So a robot can find a suitable place on $H_{i+1}$ and move there without collision. Similar care has been considered for when there is a scarcity of robots below the $H_{i}$. Hence we have downward conditions. (D2) is not exactly similar to (U2) because a priority has been given to the downward movement of robots in order to avoid any deadlock.

\section{Adoption of the proposed algorithm for finite grid}

Here we discuss how the algorithm can be modified to adopt a finite grid scenario. Our algorithm exploits the infinite grid when it asks the tail to expand the SER. If the grid is finite, then the tail may reach a corner of the grid and cannot expand the SER anymore. In this case, tail can change its color to another color say, \texttt{tailCorner}. When head robot sees this color, it starts expanding the SER. If the dimension of the grid is large enough, then after a finite number of moves by the head robot, $C_{enough}$ will be true. Then if $C_{rect}$ is not true, then also similarly, head robot can take over to make $C_{rect}=$ true.

\end{document}